\documentclass[11pt]{amsart}
\usepackage{amsfonts,amscd,amsthm,amsgen,amsmath,amssymb, hyperref}
\usepackage[all]{xy}
\usepackage[vcentermath]{youngtab}
\usepackage[letterpaper,hmargin=1.55in,vmargin=1.2in]{geometry}
\usepackage{graphicx, xcolor}
\usepackage{bm}	
\usepackage{amsmath}
\usepackage{amsthm}	
\usepackage{amsfonts}	
\usepackage{amssymb}
\usepackage{mathrsfs}
\usepackage{latexsym}
\usepackage{enumerate}
\usepackage{dsfont}
\usepackage{amssymb}

\usepackage{color}
\usepackage{soul}

\theoremstyle{plain}
\newtheorem{theorem}{Theorem}[section]

\theoremstyle{definition}

\theoremstyle{remark}

\numberwithin{equation}{section}
\numberwithin{figure}{section}

\begin{document}

\title[Non-associative magnetic translations]{Non-associative magnetic translations from parallel transport in projective Hilbert bundles}

\author[Jouko Mickelsson]{Jouko Mickelsson}
\address[Jouko Mickelsson]{Department of Mathematics and Statistics, University of Helsinki}
\email{jouko@kth.se}

\author[Michael Murray]{Michael Murray}
\address[Michael Murray]{School of Mathematical Sciences, University of Adelaide}
\email{michael.murray@adelaide.edu.au}

\maketitle
\begin{abstract} 

The non-associativity of translations in a quantum system with magnetic field
background has received renewed interest in association with topologically trivial
gerbes over $\mathbb{R}^n.$ The non-associativity is described by a 3-cocycle
of the group $\mathbb{R}^n$ with values in the unit circle $S^1.$ 
The gerbes over a space $M$ are topologically classified by the Dixmier-Douady
class which is an element of $\mathrm{H}^3(M,\mathbb{Z}).$ However, there is a finer
description in terms of local differential forms of degrees $d=0,1,2,3$ and
the case of the magnetic translations for $n=3$ the 2-form part is the magnetic field
$B$ with non zero divergence.  In this paper we study a quantum field theoretic 
construction in terms of $n$-component fermions on  a circle.
The non associativity arises when trying to lift the translation group action
on the 1-particle system to the second quantized system. 


\thanks{The second author acknowledges support under the  Australian
Research Council's {\sl Discovery Projects} funding scheme (project number DP180100383).}

\end{abstract}

\section{Introduction}

The motivation for the present short note is to understand the recent paper by
Bunk, M\"uller and Szabo \cite{BMS} in terms of quantization of Dirac operators
on a real line or on the circle coupled to an abelian vector potential with gauge
group $\mathbb{R}^n$ or the torus $T^n=\mathbb{R}^n/\mathbb{Z}^n.$
The central topic in \cite{BMS} is a 3-cocycle on $\mathbb{R}^n$ arising from composing certain
functors coming from translations acting on differential data of a topologically
trivial gerbe on $\mathbb{R}^n.$ The non-associativity in the case of a magnetic field with
sources in the case $n=3$ was suggested already long ago in \cite{Jac}. An interpretation
of the 3-cocycle in terms of representations of canonical anticommutator algebras
was then proposed in \cite{Car}.

In this paper we interpret the magnetic translations as (non periodic) $\mathbb{R}^n$ valued  gauge transformations on the unit interval $[0,1]$ acting on fermions with $n$ complex components. They actually define true operators on the level
of 1-particle Dirac operators. However, they cannot be lifted to unitary
operators in the fermionic Fock space; if they could, there would be no 3-cocycle
since the composition of linear operators is associative. Nevertheless, these
gauge transformations define functors acting on certain categories of representations
of canonical anticommutation relations. The composition of functors respects
the group law in $\mathbb{R}^n$ only modulo the action of automorphisms in the
Fock space; these automorphisms come from a projective representation of an abelian
gauge group. 

The calculation of the group 3-cocycle is based on the interpretation of the (non periodic) gauge transformations
as parallel transport in a projective  Hilbert bundle of  fermionic Fock spaces parametrized by
a simply connected Lie group $G;$ in the case of the application to magnetic translations this group
is $\mathbb{R}^n.$ This method can be viewed as a geometrization of the earlier work  \cite{Mi09}.

Let $G$ be a  simply connected Lie group.  Let $H$ be the space of square integrable periodic functions
on the unit interval $[0,1]$  taking values in the complex vector space $\mathbb{C}^n$ with a
unitary $G$ action through a representation $\rho$ of $G.$  For $g\in G$ let $P_g$ be the category (actually a Frechet manifold) of smooth
paths $f(t)$ with $0\leq t \leq 1$ in $G$ starting from the identity $e$ in $G$ and with the end point at $g,$ with the additional condition
that $A= f^{-1} df$ is smooth and periodic.   Then the set $P = \{P_g|g\in G\}$ defines a principal bundle over $G$ with fibers diffeomorphic to
the based loop group $\Omega G.$ 

Each $f\in P$ defines an 1-dimensional antihermitean Dirac operator $D_f = \frac{d}{dt} + A$ with domain $H_0   \subset H$ consisting of smooth periodic
functions.  Since  $H$ can be naturally identified as the space of square integrable functions on the unit circle we can define a polarization $H= H_+ \oplus H_-$
to positive and non-positive Fourier modes. This choice defines a fermionic Fock space $\mathcal{F}$ carrrying an irreducible representation of the
canonical anticommutation relations (CAR) algebra; this algebra is generated by elements $a^*(v), a(v)$ with $v\in H$ such that the annihilation operators
$a(v)$ depend antilinearly on the argument whereas the creation operators depend linearly on $v.$ The only nonzero anticommutators are
$$ a^*(v) a(u) + a(u) a^*(v) = <u,v> \mathbf{1}$$
where $<\cdot, \cdot>$ is the inner product in the complex Hilbert space. The representation in $\mathcal{F}$ is characterized by the existence of a cyclic
vector $ \psi_0,$ the vacuum vector, such that
$$a(u) \psi_0 = 0 = a^*(v) \psi_0,   \text{ for $u\in H_+$ and $v\in H_-$}. $$
As a reference to CAR representations and their automorphisms see \cite{Ar}. 

The category of the Dirac operators $D_f$ can be quantized (using the standard normal ordering associated to the Fourier modes in $H$) to define
quantum Dirac operators $\hat{D }_f$ acting on the dense subspace of the Fock space consisting of polynomials of the creation and annihilation
operators acting on the vacuum vector. The based loop group $\Omega G$ acts  in $\mathcal{F}$ through a central extension $\widehat{\Omega G }.$
The Lie algebra $\hat{\mathfrak{g}}$  of the central extension is defined by the 2-cocycle
\begin{equation} c_2(X,Y) = \frac{1}{2\pi i } \int  \text{tr} X dY \label{2cocycle} \end{equation}
with values in the group of purely imaginary numbers, where the trace is evaluated in the representation of $G$ in $\mathbb{C }^n.$  The 2-cocycle defines a left invariant
form on the loop group with periods  in $2\pi \mathbb{Z}$. Elements $h\in \Omega G$ define automorphisms of the category of Dirac operators
$D_f,$ parametrized by paths  $f$ with $f(1)= g,$ through right multiplication $f\mapsto fh$ which corresponds to the gauge transformation $A\mapsto h^{-1} A h
+ h^{-1}dh.$ Using the action of the central extension of the loop group in $\mathcal{F}$ this action can be lifted to the quantized Dirac operators through
$\hat h^{-1} \hat{D}_f \hat h = \hat{D}_{fh}.$ The projective representation of $\Omega G$ can be viewed as a homomorphism of $\Omega G$ to the projective unitary group
$PU(\mathcal F)$ of the Hilbert space $\mathcal F,$ so the principal bundle $P$ defines a $PU(\mathcal F)$ bundle over $G$. Projective Hilbert bundles are classified by their
Dixmier-Douady class $\omega \in  \mathrm{H}^3(G,\mathbb Z)$. The class $\omega$ in this case depends on the representation of $G$ in $\mathbb{C}^n.$ In the case of $G=SU(n)$ 
in its defining representation this class is the generator in $\mathbb{Z}=  \mathrm{H}^3(G,\mathbb{Z}).$ For $G= SU(n)$ an explicit representation of the generator as a closed 3-form is given as
$$ \frac{1}{24\pi^2} \text{tr} (g^{-1}dg)^3  \label{3-form}$$
with the trace evaluated in the defining representation of $SU(n).$  See \cite{CMM}, \cite{CMM2} for more 
information on gauge group cocycles in relation to quantized Dirac operators.

 \section{Transgression from the Dixmier-Douady class and the 3-cocycle}

For a given $g$ the elements $f\in P_g$ define a category of 1-particle Dirac operators associated to potentials $A=f^{-1} df$ for $f\in P_g$ and the quantized Dirac operators $\hat D_f.$ 
If $g,g' \in G$ is a pair of different group elements  and $f\in P_g, f' \in P_{g'}$ then \it formally \rm  $D_f$ is equivalent to $D_{f'}$ through a gauge transformation $D_f\mapsto D_{f'}$ 
by $f\mapsto fh$ with $h:[0,1]  \to G,$  $h(t) = f^{-1} f'.$ But $h$ is nonperiodic and therefore it does not preserve the domain of the Dirac operators. A consequence of this is that there is
no quantization $\hat h$ acting in the Fock space which would take $\hat D_f$ to $\hat D_{f'}$ by conjugation. Thus we can only say that the path $h$ with end point $g^{-1}g'$ takes
one category of Dirac operators (parametrized by $g$) to another category (parametrized by $g'$).

Next fix a connection in the principal bundle $P\to G.$  A based loop  $\gamma$ at $e\in  G$ defines through parallel transport an element $h(\gamma)$ in the model fiber
$\Omega G$ and thus an element in $PU(\mathcal F).$ We have now a map $\phi: \Omega G \to PU(\mathcal F).$ The canonical circle bundle $S^1 \to U(\mathcal F) \to PU(\mathcal F)$
with curvature $F$
pulls back to a circle bundle over $\Omega G$ and the Chern class of this circle bundle is represented by a  closed 2-form $\theta=\phi^* F.$

\begin{theorem} The Chern class $\theta$ on $\Omega G$ is equal to the transgression of the Dixmier-Douady class $\omega$ on $G,$ that is,
$$\theta(X,Y) = \int_t \omega (h^{-1}dh, X,Y).$$
In particular, since $\mathrm{H}^3(G,\mathbb{Z}) = \mathbb{Z}$ when $G$ is a simple compact Lie goup, we get
$$\theta(X,Y) = \frac{k}{8 \pi^2} \int_t \mathrm{tr}\, h^{-1}dh [X,Y].$$
The form $\theta$ is cohomologous to the form $\frac{i}{2\pi}c_2,$    namely $\frac{i}{2\pi} c_2= \theta + d\xi$ with $\xi(X) =\frac{1}{8\pi^2} \int_t \mathrm{tr}\, h^{-1}dh X.$
\end{theorem}

\begin{proof}  
Recall that there is a map $\tau \colon P^{[2]} \to \Omega G$ defined by $p_1 \tau(p_1, p_2) = p_2$. 
The  induced $PU(\mathcal F)$ bundle associated to $P$ by the homomorphism $\phi$ gives rise to a lifting bundle gerbe $(Q, P)$ over $G$
where  $Q$ is the pullback to $P^{[2]}$ of the bundle $U(\mathcal F) \to PU(\mathcal F) $ by $\phi \circ \tau \colon P^{[2]} \to PU(\mathcal F)$. See \cite{Mur} 
for further details such as the bundle gerbe multiplication.  This bundle gerbe has Dixmier-Douady class $\omega$ and we can choose a bundle 
gerbe connection on $Q$ with curvature $R \in \Omega^2(P^{[2]})$ and curving $B \in \Omega^2(P) $  such that $R = \pi_1^*(B) - \pi_2^*(B) $
and $dB = \pi^*(\omega)$.  From the second equation we see that the transgression of $\omega$ is $B_{| \Omega G}$ where $\Omega G \subset P$ as the space of loops. 

Let  $P_0  \subset P^{[2]}$ be defined by  $P_0 = \{(e, \gamma) \mid \gamma \in \Omega G \}$ where $e$ is the constant identity loop. Clearly $P_0$ is diffeomorphic to $\Omega G$
and $\pi_1 \colon P^{[2]} \to P$ restricted to $P_0$ induces the identity map from  $\Omega G$ to the copy of $\Omega G$ in $P$. 
On the other hand $\pi_2$  restricted to $P_0$ is the constant map to $e \in P$.  Notice also that 
$\tau$ restricted to $P_0$ is the identity.    Note that in general the pullback 
connection by $\phi \circ \tau$ might not be a bundle gerbe connection as it might fail to respect the bundle 
gerbe multiplication but certainly the Chern class of the pullback bundle is the 
pullback of the Chern class by naturality.  So the class of $\theta$ is the class of $R = \pi_1^*(B) - \pi_2^*(B)$ restricted to $P_0 \simeq \Omega G$  which is the class of $B_{| \Omega G}$ as required. 

\end{proof} 

Let $B$ be a local potential of the Dixmier-Douady form, $\omega = dB;$ this can be defined in an open neighborhood of the unit in $G.$ In the special case $G=SU(2)$ it is well-defined
in the open set not containing the matrix $-1.$ In the same way as the curvature on the loop space is a transgression of $\omega,$ we may choose a local connection $A$ in the loop space
as a transgression of $B.$ On the other hand, since the curvature in the loop space is a pull-back $\phi^*F$ we can set $A= \phi^* \eta$ where $\eta$ is a local connection form $d\eta = F,$ 
wth respect to a local trivialization of $U(\mathcal F) \to PU(\mathcal F).$ With respect to this local trivialization we can fix a phase of the lift of $\phi(\gamma)\in PU$ to the unitary group $U$
by selecting a path $\gamma_s$   ($0\leq s \leq 1$) in the loop space with $\gamma_0$ the constant loop at $e\in G$ and $\gamma_1 = \gamma.$

Next choose a triple $g_1, g_2, g_3$ of elements in $G.$  For each $g_i$ choose a path $g_i(t)$ with $g_i(0) =e$ and $g_i(1) = g_i.$ We have now four closed 1-simplices 
as follows: The first is the loop $\ell(g_1,g_2)$ defined as the composition of $g_1(t), g_1 g_2(t)$ and $-(g_1g_2 (t))=
(g_1g_2)(1-t).$  (The minus sign means the opposite orientation as compared to the path $(g_1g_2)(t).$ The second and third are $\ell(g_2,g_3)$ and $\ell(g_1,g_3).$
finally the fourth is the left translated simplex $g_1\ell(g_2,g_3).$ Then it is easy to check that as a 1-simplex
\begin{equation} \ell(g_1,g_2) -g_1\cdot \ell(g_2,g_3) + \ell(g_1g_2,g_3) - \ell(g_1, g_2g_3) = 0. \label{1-cycle}
\end{equation}
For each of the 1-simplices above choose a singular 2-simplex $s(g_1, g_2) \dots$ such that $\delta s(g_1,g_2) = \ell(g_1,g_2)$ and so on; this is possible since $G$ is simply connected.
The sum of these 2-simplices is closed by the equation above and we may choose a 3-simplex $\Delta_3(g_1,g_2,g_3)$ such that its boundary is the sum $S(g_1,g_2,g_3)$  of  the 2-simplices $s.$ 

Given four elements $g_i\in G$ with $1\leq i\leq 4$ we can set
\begin{multline}
V(g_1, \dots, g_4) = \Delta_3(g_1,g_2,g_3)  + \Delta_3(g_1g_2, g_3, g_4) +\Delta_3(g_1, g_2g_3,g_4) \\+ \Delta_3(g_1,g_2,g_3g_4) + g_1\cdot \Delta_3(g_2,g_3,g_4)
\end{multline}
and then $\delta V(g_1,g_2,g_3,g_4) =0$ since the boundaries of the $\Delta$'s cancel pairwise. From this follows that
$$C_3(g_1,g_2,g_3) = e^{2\pi i \int_{\Delta_3(g_1,g_2,g_3)} \omega}$$
is a 3-cocycle in the sense that
$$C_3(g_1,g_2,g_3) C_3(g_1g_2,g_3,g_4)^{-1} C_3(g_1, g_2g_3, g_4)C_3(g_1,g_2,g_3g_4)^{-1} C_3(g_2,g_3,g_4) = 1.$$
We have used the fact that $\omega$ is an integral form, so its integal $I$ over the 3-simplex without boundary is in $\mathbb{Z}$ and so $\exp(2\pi i I)=1.$ 

In a local trivialization of the pull-back of the bundle $U(\mathcal F) \to PU(\mathcal F)$ to the loop group $\Omega G$ the phase of the unitaries can be fixed by a  parallel transport 
along a path $\ell_s$ in $\Omega G$ connecting the identity $\ell_0=1$ to a given loop $\ell= \ell_1.$  In particular, the phases of the parallel transports around the 1-simplices
$\ell(g_i,g_j)$ are detetermined  as $\exp2\pi (i\int_{s(g_i,g_j)} B)$ since the connection in the loop space is the transgression of the local form $B.$ The phases of the parallel
transports along the faces of the 3-simplex $\Delta_3(g_1,g_2,g_3)$ add up to $ \exp(2\pi i \int_{\partial \Delta_3(g_1,g_2,g_3)} B)$ which is equal to $\exp(2\pi i\int_{\Delta_3(g_1,g_2,g_3)}  \omega)$
by Stokes' theorem. Thus we can state
\begin{theorem} The obstruction to lifting the 2-cocycle  $\ell$ to the central extension $U(\mathcal F)$ of $PU(\mathcal F)$ is   the class of 3-cocycle $C_3.$
Changing the cocycle $C_3$ by a coboundary $\delta b$ is the same as changing the lifts of the loops $\ell$ to the central extension by phases $b(g_i,g_j).$ 
\end{theorem}

An equivalent way to arrive at the 3-cocycle $C_3$ is as follows.
Denoting the parallel transports along the edge loops  $ \ell$  by the same symbols we have
then the cocycle property (which follows from \ref{1-cycle} )
$$ \ell(g_1, g_2)\ell(g_1g_2,g_3) =    \ell(g_2, g_3)^{g_1}\ell(g_1,g_2g_3)$$
where the first factor on the right denotes the parallel transport  from $g_1$ to $g_1g_2$ and continuing via $g_1g_2g_3$ back to the vertex $g_1.$

Next we fix a lift $\hat{\ell}$ to the central extension $\widehat{\Omega G}$ acting in the Fock space $\mathcal{F}_0.$ The lift is fixed by a parallel transport from the constant loop
to the loop $\ell$ as explained previously. The cocycle property above does not hold anymore for the lifts but
there is a correction factor,
$$\hat{\ell}(g_1, g_2) \hat{\ell}(g_1g_2, g_3) = \hat{\ell}(g_2, g_3)^{g_1} \hat{\ell}(g_1, g_2g_3)C_3(g_1,g_2,g_3),$$
with the   $S^1$ valued 3-cocycle $C_3.$

The cocycle $C_3$ depends on the choices of the lifts $\hat{\ell}$ modulo a coboundary. However, since the loops involved come from the boundaries of the faces in a tetraed we can
use this to fix the phases. The central extension $\widehat{\Omega G}$ comes equipped with a canonical connection: The tangent bundle of a Lie group is trivial and the tangent spaces
can be identified as the Lie algebra of the group. The Lie algebra of the central extension is as a vector space the direct sum of the Lie algebra of $\Omega G$ and the 1-dimensional
Lie algebra  $i\mathbb{R}. $ The projection onto the center defines the connection 1-form $\xi$ on $\widehat{\Omega G}.$  This connection is actually the connection coming
from the map $\Omega G \to PU(\mathcal F)$ through canonical quantization \cite{Lu}. The face maps to $G$ define a contraction of the parallel
transport along each  boundary loop
to a constant loop at $e.$ The contraction defines a path in $\Omega G$ starting from the constant loop and using the canonical connection this defines an element in $\widehat{\Omega G}.$
If $\delta,\delta'$ are a pair of face maps with a common boundary 1-simplex then the phases differ by a factor $\psi(\delta,\delta')\in S^1= \mathbb{R}/\mathbb{Z}.$ Since $\mathrm{H}^2(G,
\mathbb{Z})=0$ the union $\delta\cup \delta'$ is a boundary of some volume $V\subset G$ and then $\psi(\delta, \delta')= \exp(2\pi i \int_V \omega).$

Returning to the calculation of the cocycle $C_3(g_1,g_2,g_3)$:  The phases of the parallel transports along the edges of the tetraed are fixed by the faces. The difference of the phases is
the given by 
$$C_3(g_1,g_2,g_3) = \exp(2\pi i\int_{\Delta_3(g_1,g_2,g_3)}  \omega).$$

As stated before in Section 1, a parallel  transport from a point $g$ to a point $g'$ can be viewed as a functor 
in the 1-particle Hilbert space $H.$ However, it does not define an operator in
a Fock space.  But the compositions of the functors along the edges of the faces of $\Delta_3$ define projective operators and missmatch of the phases of the composed lifted unitary operators
give rise to the $S^1$ valued 3-cocycle.

In the case of a topologically trivial gerbe, $\omega = dB$ for some globally defined 2-form on $G,$ the parallel transport in the loop space is defined by the transgression of $B$
giving the potential in the loop space as
$$ A(f; X) = \int_{S^1}  B(f^{-1}df, X)$$
for a loop $f$ and a tangent vector $X$ at  $f.$

That gives the parallel transport from the constant loop at $e$ to the boundary loop of the 2-simplex $s(g_1, g_2)$ (with vertices $e, g_1, g_1g_2$) as
$$d(g_1, g_2) = \exp( 2\pi i\int_{s_2} B).$$
In this case the 3-cocycle $C_3$ is trivial,
$$C_3(g_1, g_2, g_3) = d(g_1,g_2) d(g_1g_2, g_3) d(g_1, g_2 g_3)^{-1} g_1\cdot d(g_2, g_3)^{-1}$$
where in the last factor the intergation is around the face on the 3-simplex with vertices at $g_1, g_1g_2, g_1g_2g_3.$

 \section{The case of $G= \mathbb{R}^3$}
 
 The group $G=\mathbb{R}^3$ has an unitary representation in the Hilbert space
 $H$ of square integrable periodic  functions on $[0,1]$ with values in $\mathbb{C}^3$ 
 through multiplication $z_k \mapsto e^{2\pi i x_k} z_k$ for $k=1,2, 3.$
 This defines also an action of the loop group $LG$ in $H$ through point-wise
 multiplication by the phase $e^{2\pi i x_k(t)}.$ 
 
 Through the canonical quantization in a fermionic Fock space the Lie algebra valued loops $X,Y: S^1 \to i{\mathbb  R}^3 $ are 
represented projectively with the 2-cocycle
$$c_2(X,Y) =\frac{1}{2\pi i}  \int_{S^1}  <X,dY>.$$
This 2-cocycle is trivial in the cohomology with coefficients in the space of Frechet smooth functions on $\Omega G,$ 
   $c_2=\delta b$ with $b_f(X) = \frac{1}{4\pi i} \int_{S^1}  <d\log (f), X>.$ 
Let $c_2'$ be another trivial cocycle,
$$c_2'(f; X,Y) = \frac{1}{2\pi i}\int_{S^1} d\log f \wedge X \wedge Y.$$ 
One can check that $c_2'$ is the trangression of of the 3-form $\omega(X,Y,Z)= X\wedge Y \wedge Z $ on $G$ 
As a 3-form $\omega$ is exact, $\omega= dB$ with $B_w(X,Y) = w\wedge X \wedge Y$ for $w\in \mathbb{R}^3$ and $X,Y$ tangent vectors at $w.$
We can take $\omega$ as the Dixmier-Douady form of topologically trivial $PU$ bundle over ${\mathbb R}^3.$ In this case the argument in the previous
section lead to a 3-cocycle of ${\mathbb R}^3,$

$$C_3(X,Y,Z) =\exp(2\pi i \int_{\Delta_3} \omega)  = \exp(\frac{2\pi i}{6} X\wedge Y \wedge Z)$$
where $\Delta_3$ is the tetraed with vertices at the points $0, X, X+Y, X+Y+Z.$
This cocycle is a coboundary of the 2-cochain $R,$ with coefficients in the group $G= \mathbb{R}^3,$ 
$$R_z(X,Y) =  \exp(\pi i \, z  \wedge X \wedge Y).$$
This is again a consequence of the fact that the Dixmier-Douady form $\omega$ is trivial in cohomology.

We can easily generalize the discussion above to the case $G=\mathbb{R}^n$ for $n\geq 3.$ 
Each antisymmetric real valued tensor $a_{ijk}$ on $\mathbb{R}^n$ defines  a closed form
 $\omega= \sum_{i,j,k} a_{ijk} dx_i \wedge dx_j \wedge dx_k$ which can be taken as the Dixmier-Douady form of a trivial
gerbe on $G.$ This closed form is exact,
$\omega = dB$ with $B= \sum_{i,j,k} a_{ijk} x_i dx_j \wedge dx_k.$ 

The cocycle $C_3$ is recovered by the parallel transport argument as before. In this case we can take the path connecting the identity $0\in\mathbb{R}^n$ to 
a vector $X$ as the straight line $X(t) = tX.$ The loops $\ell$ are now triangles connecting the vectors $0,X, X+Y$, $0,X, X+Y+Z$, $0, X+Y, X+Y+Z$ and $X, X+Y, X+Y+Z$
and the cocycle $C_3(X,Y,Z)$ is again given by  the integral  of the 3-form $\Omega$ over  the 3-simplex with vertices at the points $0,X,X+Y, X+Y+Z,$ 
\begin{equation} C_3(X,Y,Z)= e^{2\pi i \sum_{ijk} \frac{1}{6} a_{ijk} X_i Y_j Z_k }.\label{3cocycle} 
\end{equation}

Although this cocycle for (nonzero $a$) is nontrivial as a group cocycle it is
however trivial as a transformation groupoid cocycle: The group $\mathbb{R}^n$ 
acts on itself by translations and $C= \delta b$ for the the 2-cochain
$b(u; X,Y) = C_3(u,X,Y)$
with
$$(\delta b)(X,Y,Z) = b(u; X,Y)^{-1}  b(u; X+Y, Z)^{-1}  b(u; X, Y+Z) b(u+X; Y,Z)$$

In the present setting the loops take values in $\mathbb{R}^n$ and $\rho(x)z_k =
e^{ix_k} z_k$ for $k=1,2,\dots, n.$ 
Each component defines a circle value function $e^{i f(t)}$ acting as a multiplication operator
in the 1-particle Hilbert space $H=L_2([0,1], \mathbb{C}^n).$ The cocycle $c_2$ \eqref{2cocycle} is nontrivial on the abelian
loop group. However, in the case of a family of Dirac operators $D_A = i \frac{d}{dt}
+ A$ coupled to an abelian vector potential $A$ (with values in $\mathbb{R}^n$) the relevant cohomology is
with coefficients in the space of Frechet smooth functions of the vector potential.  In this cohomology the cocycle becomes trivial:
we have $c_2 = \delta b_1$ where 
$$b_1(A; X) = \frac{1}{4\pi i} \int \sum_k A_k X_k dt$$
and the loop algebra element $X$ acts on $A$ through the gauge transformation $A \mapsto A + dX.$ 
For this reason the bundle of Fock spaces parametrized by the vector potentials
becomes equivariant with respect to the gauge action and can be pushed forward
to a bundle over the flat moduli space $\mathbb{R}^n= \mathcal{A}/\mathcal{G}$ of gauge potentials;
here $\mathcal{G}$ is the group of periodic functions  $[0,1] \to \mathbb{R}^n$ 
acting on potentials as 
$A\mapsto A+  df.$ 

\section{The case of a torus}

If we replace the gauge group $\mathbb{R}^n$ by the torus $T^n$ the situation
becomes different. All the maps $f:[0,1] \to \mathbb{R}^n$ which are periodic
modulo $\mathbb{Z}^n$ satisfy the Hilbert-Schmidt condition on off-diagonal blocks
with respect to the energy polarization; again, a function $f$ defines a multiplication
operator in the one-particle space through $z_k \mapsto e^{2\pi i  f_k}z_k.$ These functions $f$ van be viewed as loops $S^1 \to T^n.$
Now the group of gauge transformations $\mathcal{G}$ factorizes as a product of the
group $\mathcal{G}_0$ contractible
maps to $T^n$ (represented as loops on $\mathbb{R}^n$) and a group $\mathbb{Z}^n$
of maps  of
the form $f(t) = 0$ for $t\leq 0,$ $f(t) = t v$ for $0\leq t\leq 1$ with $v \in \mathbb{Z}^n$ and $f(t) =v$ for $t\geq 1.$ 

The moduli space of gauge potentials $\mathcal{A}/\mathcal{G}$ is now
the torus $T^n;$ we have $\mathcal{A}/\mathcal{G}_0 = \mathbb{R}^n$ and the second
factor in $\mathcal{G}$ is isomorphic to the subgroup $\mathbb{Z}^n \subset \mathbb{R}^n.$ In the case of $\mathbb{R}^n$ there was no restriction on the normalization of
the 3-cocycle [as a group cocycle or as a 3-form on $\mathbb{R}^n$] but in the case
of the torus the 3-cocycle must satisfy an integrality constraint in order that
the gerbe over $T^n$ is well-defined.

As explained in \cite{HaMi}, \cite{Mi17} (see also \cite{MiWa} Section 7) the 1-particle  Dirac hamiltonians can
be twisted in such a way that their K-theory class over the moduli space $T^n$ 
is nontrivial: the Chern character has a nonzero component $\omega_3$ in $\mathrm{H}^3(T^n, \mathbb{Z}).$ The basis in $\mathrm{H}^3(T^n,\mathbb{Z})$ is given by the
3-forms $\omega = \sum a_{ijk} dx_i \wedge dx_j\wedge dx_k$ where the $a's$ form
a basis of totally antisymmetric tensors of rank 3 with integral coefficients,
The pull-back with respect to the projection $\mathbb{R}^n \to T^n$ is the form
$d \sum a_{ijk} x_i dx_j\wedge dx_k.$ The quantum field theoretic construction
of a gerbe over the torus from a non zero class $[\omega]$ is recalled in the
Appendix. 

The 3-form part $\omega$ of the Chern character is the Dixmier-Douady class of
the projective vector bundle over $T^n$ obtained by canonical quantization of the
family of 1-particle Dirac operators. The pull-back of this bundle over $\mathbb{R}^n$
comes by projectivization of a vector bundle (the bundle of fermionic Fock spaces).
The group $\mathbb{Z}^n$ acts through an abelian extension of the Fock spaces.
The extension is defined by the 2-cocycle
$$ c_2(u; x, y)= e^{2\pi i \sum a_{ijk} u_i x_j y_k}$$
where $x,y\in \mathbb{Z}^n$ and $u\in\mathbb{R}^n$ and $\mathbb{Z}^n$ acts on
the functions of the vector $u$ as translations.

For integral coefficients $a_{ijk}$ the 3-cocycle \eqref{3cocycle}  is identically $=1$ when the arguments are in $\mathbb{Z}^n
\subset \mathbb{R}^n$ in conformity with the (projective) action of $\mathbb{Z}^n$
on the Fock spaces. The 3-cocycle is the obstruction to an extension of the $\mathbb{Z}^n$ action
to a $\mathbb{R}^n$ action on the bundle of Fock spaces.

{\bf Remark} The cocycle $c_2$ is also a group cocycle even in the case of constant
coefficients (no group action on $u$) but since the coboundary operator is
different  the cohomology with variable coefficients is different from the cohomology
with constant coefficients.

\section{Appendix} 
In this appendix we briefly recall the quantum field theoretic construction of a
gerbe over the torus using a twisted family of CAR algebra representations,
\cite{HaMi}, \cite{Mi17}.

A hermitean complex line bundle $L$ over the torus $T^n$ is characterized by a class $\omega$ in
$\mathrm{H}^2(T^n, \mathbb{Z}).$ Parametrizing the circles in the torus by the interval
$[0,1]$  the 2-cohomology is spanned  by antisymmetric bilinear forms on $\mathbb{R}^n$ such that $\omega(x,y) \in \mathbb{Z}$ for $x,y\in
\mathbb{Z}^n.$ The pull-back of $L$ over $\mathbb{R}^n$ is trivial and
the sections of that line bundle are complex valued functions $\psi$ such that
$$\psi(x + z) = \psi(x) e^{2\pi i \omega(x,z)}$$
for $z\in \mathbb{Z}^n.$

Next we construct a family of fermionic Fock spaces parametrized by vectors in
$\mathbb{R}^n.$ For each $k=1,2, \dots, n$ and $u,v\in H$ let $a_k(v), a^*_k(u)$
be generators of a CAR algebra with nonzero anticommutators
$$a^*_k(u) a_k(v) + a_k(v) a^*_k(u) = 2 <u,v>.$$
The generators for different lower indices are assumed to  commute. It will be
convenient to compactify the real line to the unit circle so we can take $H=
L_2(S^1, \mathbb{C}^n)$ and we can work with the orthonormal basis of Fourier modes
in each of the $n$ directions.

We twist the CAR algebra by the line bundle $L.$ This means that the families of creation
and annihilation operators are sections of the tensor products of $L$ or its dual and the CAR algebra. 
The sections are $\mathbb{Z}^n$ equivariant functions on $\mathbb{R}^n,$ 
that is, for  $x\in\mathbb{R}^n$  and  for $z\in\mathbb{Z}^n$
$$ a^*_k(u, x+z) = a^*_k(u,x) e^{2\pi i \omega(x,z)}, \,\, a_k(u,x+z) = a_k(u,x) e^{-2\pi i \omega(x,z)}.$$
The right-hand-side of the canonical anticommutation relations, when evaluated at a point
$x\in\mathbb{R}^n,$ is multiplied by the pairing of sections of $L,L^*$ involved in the
construction of $a^*_k(u,x) = a^*_k(u)\otimes \psi(x)$ and of $a_k(v)\otimes \xi.$

The  Fock vacuum is again annihilated by $a^*(u,x)$ and $a(v,x)$ for $u\in H_-$ and
$v\in H_+.$ (One could generalize this construction by allowing the modes for 
different lower index $k$ be twisted by different line bundles.) 

Thus the states with net particle number $N$ in the Fock space are twisted by 
the $N$:th tensor power of $L.$  

The group $\mathbb{Z}^n$ acts as automorphisms of the twisted  CAR algebra by
$$g(p) a^*(u,x) g(p)^{-1}= a^*(p\cdot u, x+p) = e^{2\pi i \omega(x,p)} a^*(p\cdot u,x)$$
where $p$ acts on a function $u(\zeta)$ by multiplication by a phase, $u_k
\mapsto e^{2\pi i \zeta p_k} u_k,$
that is, the Fourier modes  are shifted by $p$ units.
Likewise, for the annihilation operators
$$g(p) a(u,x) g(p)^{-1} = e^{-2\pi i \omega(x,p)} a((-p)\cdot u, x).$$
The action of $g(p)$ in the Fock spaces parametrized by $x$ is now completely defined by
fixing the action on the vacuum vector $\psi.$ This is easiest done thinking the vectors
as elements in the semi-infinite cohomology (in physics terms, the \lq Dirac sea').
For $n=1$ the vacuum is symbolically the semi-infinite  product
$$ \psi = a^*_0 a^*_{-1} a^*_{-2} \cdots$$
where the lower index refers to the Fourier modes in $L_2(S^1).$ For general $n$ the
vacuum is defined in a similar way inserting the non-negative Fourier modes for each
of the $n$ components.  The CAR generators are labelled by a double index $(k,j)$
with $k\in\mathbb{Z}$ and $j=1,2,\dots n. $ The action of $g(p)$ on the vacuum is now defined as a shift
operator: The index $k$ of the element $a^*_{k,j}$  is shifted by the integer  $p_j$ for $j=1,2, \dots, n,$ $k\mapsto k+ p_j.$ 
 
Because of the phase shifts when the CAR algebra generators are conjugated by $g(p)$
the product $g(p) g(q)$ is not equal to $g(p+q)$ but they differ by an $x$ dependent phase,
$$g(p)g(q) = C(x; p,q) g(p+q)= e^{2\pi i N \omega(x,p)} g(p+q)$$
where $N$ is the particle number of the state $g(q)\psi,$ that is, $N= \sum_{j=1}^{j=n} q_j.$

{\bf Remark} The projective vector bundles over $T^n$ are classified by elements of
$\mathrm{H}^3(T^n, \mathbb{Z}.$ Representatives of these elements can be written as
de Rham forms $\Omega = \sum a_{ijk} dx_ \wedge dx_j \wedge dx_k$ where the coefficients
$a_{ijk}$ are integers. The pull-back of $\Omega$ with respect to the projection
$\pi: \mathbb{R}^n \to T^n$ is $\pi^*\Omega = d\theta= d\sum a_{ijk} x_i dx_j \wedge dx_k.$
Evaluating $\theta$ for tangent vectors $u,v$ in the integral lattice $\mathbb{Z}^n
\subset \mathbb{R}^n$ and exponentiating gives the 2-group cocycle
$$C'(x; u,v) = e^{2\pi i \sum a_{ijk} x_i u_j v_k}$$
where the group $\mathbb{Z}^n$ acts on the vector $x$ by $x\mapsto x+u.$
According to the discussion in \cite{MiWa}, Section 7.1, there is a 1-1 correspondence
between the group cohomology $\mathrm{H}_{grp}^2(\mathbb{Z}^n, A)$ and the de Rham
cohomology $\mathrm{H}^3(T^n, \mathbb{Z})$ where $A$ is the $\mathbb{Z}^n$ module
of (smooth) functions $T^n\to S^1.$ The above map $\{c_{ijk}\} \to C'$
realizes this isomorphism.

{\bf Example}: When $n=3$ the cocycle $C$ is equivalent to $C'$ for the choice
$a_{ijk} =\alpha \epsilon_{ijk}$ where 
where $\epsilon$
is totally antisymmetric tensor with $\epsilon_{123}=1$ and $\alpha = 2(\omega_{12}
+\omega_{23} + \omega_{31})$ where $\omega= \omega_{12} dx_1\wedge dx_2 + \omega_{31}
dx_3\wedge dx_1 + \omega_{23}dx_2\wedge dx_3.$
This is seen by projecting the exponent in $C$ to its totally antisymmetric
component.

 \end{document}